\documentclass[amslatex,12pt]{amsart}
\usepackage{enumerate,setspace}
\usepackage{amsmath,amssymb,amsfonts,ascmac}
\usepackage{amsthm}
\usepackage{graphicx}
\usepackage{algorithm}
\usepackage{algorithmic}
\usepackage[usenames]{color}
\usepackage{multirow}
\usepackage{tikz}
\usetikzlibrary{arrows,shapes, calc, fit, positioning}
\usepackage{tkz-graph}
\newcommand{\calF}{\mathcal{F}}
\newcommand{\calP}{\mathcal{P}}
\newcommand{\calC}{\mathcal{C}}
\newcommand{\calN}{\mathcal{N}}
\newcommand{\calG}{\mathcal{G}}
\newcommand{\sfel}{\mathsf{el}}
\newcommand{\sfl}{\mathsf{\ell}}
\newcommand{\sfc}{\mathsf{c}}
\newcommand{\sfr}{\mathsf{r}}
\newcommand{\sfer}{\mathsf{er}}
\newcommand{\bbR}{\mathbb{R}}

\newcommand{\height}{\mbox{\rm height}}
\newcommand{\suc}{\mbox{\rm succ}}
\newcommand{\best}{\mbox{\rm B}}
\newcommand{\ch}{\mbox{\rm ch}}
\newcommand{\pr}{\mbox{\rm pr}}
\newcommand{\msucceq}{\mathop{\stackrel{m}{\succeq}}}

\newcommand\DoubleLine[5][3pt]{%
  \path(#2)--(#3)coordinate[at start](h1)coordinate[at end](h2);
  \draw[#4]($(h1)!#1!90:(h2)$)--($(h2)!#1!-90:(h1)$);
  \draw[#5]($(h1)!#1!-90:(h2)$)--($(h2)!#1!90:(h1)$);
}

\newtheorem{theorem}{Theorem}
\newtheorem{lemma}[theorem]{Lemma}
\newtheorem{proposition}[theorem]{Proposition}
\newtheorem{cor}[theorem]{Corollary}
\newtheorem{definition}{Definition}
\newtheorem{example}{Example}

\numberwithin{equation}{section}
\title{Hedonic Games with Graph-restricted Communication}

\author[]{Ayumi Igarashi}
\address[Ayumi Igarashi]{Department of Computer Science, University of Oxford, U.K.}
\email[]{ayumi.igarashi@cs.ox.ac.uk}

\author[]{Edith Elkind}
\address[Edith Elkind]{Department of Computer Science, University of Oxford, U.K.}
\email[]{elkind@cs.ox.ac.uk}
\date{}

\begin{document}
\maketitle

\begin{abstract}
We study hedonic coalition formation games in which cooperation among the players is restricted by a graph 
structure: a subset of players can form a coalition if and only if they are connected in the given 
graph. We investigate the complexity of finding stable outcomes in such games, for several notions
of stability. In particular, we provide an efficient algorithm that finds an individually stable
partition for an arbitrary hedonic game on an acyclic graph.
We also introduce a new stability concept---in-neighbor stability---which is tailored for our setting.
We show that the problem of finding an in-neighbor stable outcome admits a polynomial-time 
algorithm if the underlying graph is a path, but is NP-hard for arbitrary trees
even for additively separable hedonic games; for symmetric additively separable games
we obtain a PLS-hardness result.
\end{abstract}

\section{Introduction}\label{sec:intro}
\noindent
In human and multiagent societies, agents often need to form coalitions in order to achieve their goals.
The coalition formation process is guided by agents' beliefs about the performance 
of each potential coalition. Many important aspects of coalition formation can be studied using the formalism of 
{\em hedonic games}~\cite{Banerjee2001,Bogomolnaia2002}. In these games, each agent has preferences
over all coalitions that she can be a part of, and an outcome is a partition of agents into coalitions.
An important consideration in this context is {\em coalitional stability}: an outcome should be resistant
to individual/group deviations, with different types of deviations giving rise to different notions
of stability (such as core stability, individual stability, Nash stability, etc.; see the recent survey
of Aziz and Savani~\cite{Aziz15} for an overview).

The standard model of hedonic games does not impose any restrictions on which coalitions may form.
However, in reality we often encounter constraints on
coalition formation. Consider, for instance, an international network of natural gas pipelines. It seems unlikely 
that two cities disconnected in the network would be able to coordinate a trading agreement without any help from  
intermediaries. Such restrictions on communication structure can be naturally described by undirected graphs,
by identifying agents with nodes, communication links with edges, and feasible coalitions with connected subgraphs. 
In the context of cooperative transferable utility games this model was proposed in the seminal paper
of Myerson~\cite{Myerson1977}, and has received a considerable amount of attention since then.
In contrast, very little is known about hedonic games
with graph-restricted communication, though some existing results for general non-transferable utility
games have implications for this setting. In particular, the famous result of Demange~\cite{Demange2004}
concerning stability in cooperative games on trees extends to non-transferable utility games, and implies
that every hedonic game whose communication structure is acyclic admits a core stable partition
(we discuss this result in more detail in Section~\ref{sec:core}). However, no attempt has
been made to obtain similar results for other hedonic games solution concepts, or to explore
algorithmic implications of constraints on the communication structure (such as acyclicity
or having a small number of connected subgraphs) for computing the core and other solutions. The goal of this paper is to make the first step towards filling this gap.

\smallskip

\noindent{\bf Our contribution\ } Inspired by Demange's work, we focus on hedonic games on acyclic graphs.
We consider several well-studied notions of stability for hedonic games, 
such as individual stability, Nash stability, core stability and strict core stability
(see Section~\ref{sec:prelim} for definitions), and ask two questions: 
(1) does acyclicity of the communication structure guarantee the existence of a stable outcome?
(2) does it lead to an efficient algorithm for computing a stable outcome, and if not, 
are there additional constraints on the communication structure that can be used to obtain such as algorithm?
We remark that, in general, to represent the preferences of a player in an $n$-player hedonic game,
we need to specify $2^{n-1}(2^{n-1}-1)/2$ values, which may be problematic if we are interested 
in algorithms whose running time is polynomial in $n$. We consider two approaches to circumvent 
this difficulty: (a) working in the oracle model, where an algorithm may submit a query
of the form $(i, X, Y)$ where $X$ and $Y$ are two coalitions that both contain $i$,
and learn in unit time whether $i$ prefers $X$ to $Y$, $Y$ to $X$ or is indifferent between them; 
(b) considering specific
succinct representations of hedonic games, such as additively separable hedonic games~\cite{Bogomolnaia2002},
which can be described using $n(n-1)$ numbers.

We observe that Demange's algorithm for the core runs in time that is polynomial in the number of connected
subtrees of the underlying graph $G$ (in the oracle model), 
and use similar ideas to obtain an algorithm for finding an outcome that is both
core stable and individually stable as well as an algorithm for finding a Nash stable outcome (if it exists).
The running time of these algorithms can be bounded in the same way; 
in particular, they run in polynomial time when $G$ is a path. However, we show that when $G$ is a star,
finding a core stable, strictly core stable or Nash stable outcome is NP-hard, even if we restrict ourselves
to very simple subclasses of additively separable hedonic games. For symmetric additively separable hedonic
games, we show that the PLS-hardness result for Nash stability \cite{Gairing2010} holds even if 
$G$ is a star.

In contrast, acyclicity turns out to be sufficient for individual stability: we show that every hedonic game
on an acyclic graph admits an individually stable partition, and, moreover, such a partition
can be computed in time polynomial in the number of players (in the oracle model). 
We believe that this result is remarkable,
since in the absence of communication constraints finding an individually stable outcome
is hard even for (symmetric) additively separable hedonic games~\cite{Sung2010,Gairing2010},
and finding a Nash stable outcome in such games remains hard even for games on stars (Section~\ref{sec:INS}).

Another contribution of our paper is a new stability concept that is tailored specifically to hedonic games
on graphs, and captures the intuition that, to join a group, a player should be approved by the members
of the group who know him. The resulting solution concept, which we call {\em in-neighbor stability},
lies between Nash stability and individual stability. However, we show that from the algorithmic
perspective it behaves similarly to Nash stability; in particular, finding an in-neighbor stable outcome
is NP-hard for additively separable hedonic games on stars and PLS-hard for symmetric additively separable hedonic 
games on stars. Our computational complexity results are summarized in Table~\ref{tab:table}.

\smallskip

\noindent{\bf Related work\ }
Sung and Dimitrov~\cite{Sung2010} were the first to consider complexity issues 
in additively separable hedonic games (ASHGs); they prove that it is NP-hard to determine if a game 
admits a core stable, strict core stable, individually stable, or Nash sable outcome
(see also~\cite{petersE}).
Aziz et al.~\cite{Aziz2013} 
extend the first two or these results to symmetric additively separable hedonic games (SASHGs).
While SASHGs always admit a Nash stable or individually stable partition~\cite{Bogomolnaia2002,burani},
finding one may still be difficult: Gairing and Savani~\cite{Gairing2010} prove that finding such partitions is 
PLS-hard (PLS-hardness is a complexity class for total search problems, see~\cite{Schaffer1991}).
The running time of some of our algorithms is polynomial in the number of connected coalitions;
see the recent work of Elkind~\cite{DBLP:conf/wine/Elkind14} for a characterization of graph families
for which this quantity is polynomial in the number of nodes. Some
papers~\cite{bilo1,bilo2,peters}
use the phrase ``hedonic game on a tree'' to refer to a hedonic game where each player $i$ has
a value $v_i(j)$ for every other player $j$, and pairs $\{i, j\}$ such that $v_i(j)\neq 0$ or 
$v_j(i)\neq 0$ form a tree; the preference relation of player $i$ is computed 
based on the values $v_i(j)$: the value of a coalition $S$, $i\in S$, could be 
$\sum_{j\in S\setminus \{i\}}v_i(j)$ (this corresponds to ASHGs) or 
$\frac{1}{|S|}\sum_{j\in S\setminus\{i\}}v_i(j)$ (such games are known as {\em fractional hedonic games})
This framework is different
from ours: we allow preferences that are not derived from values assigned to individual players, 
and in the additively separable case we allow non-adjacent players to have a non-zero 
value for each other. 

\begin{table}
\begin{center}
{\tabcolsep=2mm
\begin{tabular}{ |c|c c c c c c| }
\hline
	& \multicolumn{2}{c}{Acyclic graphs} 	 & \multicolumn{3}{c}{Stars} &Paths \\ 
\hline
	&Arbitrary & Additive	& Additive	& S-additive 	& S-enemy &	 Arbitrary\\ 
\hline
SCR 	& - 		&NP-h  &NP-h 		& NP-h (Th.\ref{thm:SCR:additive})		& NP-h$^*$ (Th.\ref{thm:SCR:enemy})	& ?\\ 
CR 	& - 		& NP-h$^*$ & NP-h$^*$ 		& NP-h$^*$ 		& NP-h$^*$ (Th.\ref{thm:CR}) 	& P (\cite{Demange2004}) \\ 
NS 	& - 		& NP-c & NP-c (Th.\ref{thm:NS})	& PLS-c (Th.\ref{thm:sym:NS}) & P (Prop.\ref{prop:sym:enemy:NS}) 	& P (Th.\ref{thm:NS_NIS_IRNIS})\\ 
INS 	& - 		& NP-c & NP-c (Th.\ref{thm:INS})	& PLS-c (Th.\ref{thm:sym:INS})	& P 		& P (Th.\ref{thm:NS_NIS_IRNIS})\\ 
IR-INS& -  	& NP-c (Th. \ref{thm:IRINS}) 		& P (Prop.\ref{prop:additive:IRINS})  	& P 		& P 		& P (Th.\ref{thm:NS_NIS_IRNIS})\\ 
IS& P (Th.\ref{thm:IS})    & P 	& P 		& P 		& P 		& P\\ \hline
\end{tabular}
}
\end{center}
\caption{Complexity of computing stable outcomes for hedonic games on acyclic graphs. 
The top row corresponds to restrictions on graphs; the second row from the top indicates restrictions on preference 
profiles. The positive results for unrestricted preferences are in the oracle model.
To avoid dealing with representation issues, 
when a problem is NP-hard for additively separable games on stars, 
we do not consider its complexity for unrestricted preferences (indicated by `-').
The hardness result marked with $^*$ holds with respect to Turing reductions.
When no reference is given, the result follows trivially from other results in the table.
The results for paths hold for all trees with $n$ nodes and $\mathrm{poly}(n)$
connected subtrees.
\label{tab:table}}
\end{table}

\section{Preliminaries}\label{sec:prelim}
\noindent
We start by introducing basic notation and definitions of hedonic games and graph theory. 
\smallskip

\noindent{\bf Hedonic games \ }
A {\it hedonic game} is a pair $(N,(\succeq_{i})_{i \in N})$ where $N$ is a finite set of players and each 
$\succeq_{i}$ is a complete and transitive preference relation over the nonempty subsets of $N$ including player $i$.
The subsets of $N$ are referred to as {\em coalitions}.
We let $\calN(i)$ denote the collection of all coalitions containing $i$.
We call a coalition $X \subseteq N$ {\it individually rational} if $X \succeq_{i} \{i\}$ for all $i \in X$. Let 
$\succ_{i}$ denote the strict preference derived from $\succeq_{i}$, i.e., 
$X \succ_i Y$ if $X \succeq_i Y$, but $Y \not \succeq_{i} X$. 
Similarly, let $\sim_{i}$ denote the indifference relation induced by $\succeq_{i}$, i.e., 
$X \sim_{i} Y$ if $X \succeq_{i} Y$ and $Y \succeq_{i} X$.

An important subclass of hedonic games is {\em additively separable games}. 
These games model situations where each player has a specific value for every other player, and 
ranks coalitions according to the total value of their members~\cite{Bogomolnaia2002}. Formally, a preference profile 
$(\succeq_{i})_{i \in N}$ is said to be {\it additively separable} if there exists a utility matrix $U:N \times N 
\rightarrow \bbR$ such that for each $i \in N$ and each $X,Y\in\calN(i)$ 
we have $X \succeq_{i} Y$ if and only if 
$\sum_{j \in X}U(i,j) \geq \sum_{j \in Y}U(i,j)$~\cite{Bogomolnaia2002}. Without loss of generality, we will assume that
$U(i,i)=0$ for each $i \in N$. An additively separable preference is said to be {\it symmetric} if the utility matrix 
$U:N \times N \rightarrow \bbR$ is symmetric, i.e., $U(i,j)=U(j,i)$ for all $i,j \in N$. 
Dimitrov et al.~\cite{Dimitrov2006} studied a subclass of additively separable preferences, 
which they called {\it enemy-oriented preferences}. Under these preferences 
each player considers every other player to be either a friend or an enemy, and has strong aversion towards her 
enemies: $U(i,j)\in \{1,-|N|\}$ for each $i,j \in N$ with $i \neq j$.

An {\em outcome} of a hedonic game is a partition of players into disjoint coalitions. Given a partition $\pi$ of $N$ and 
a player $i \in N$, let $\pi(i)$ denote the unique coalition in $\pi$ that contains $i$.
The first stability concept we will introduce is {\it individual rationality}, 
which is often considered to be a minimum requirement that solutions should satisfy. 
A partition $\pi$ of $N$ is said to be {\it individually rational} if all 
players weakly prefer their own coalitions to staying alone, i.e.,  $\pi(i) \succeq_{i} \{i\}$
for all $i\in N$.

The {\em core} is one of the most studied solution concepts in hedonic games \cite{Dreze1980,Banerjee2001,Bogomolnaia2002}.
A coalition $X \subseteq N$ {\it strongly blocks} a partition $\pi$ of $N$ if $X \succ_{i} \pi(i)$ for all $i \in X$; 
it {\it weakly blocks} $\pi$ if $X \succeq_{i} \pi(i)$ for all $i \in X$ and $X \succ_{j} \pi(j)$ for some $j \in X$.
A partition $\pi$ of $N$ is said to be {\it core stable} (CR) if no coalition $X \subseteq N$ 
strongly blocks $\pi$; it is said to be {\it strictly core stable} (SCR) 
if no coalition $X \subseteq N$ weakly blocks $\pi$.

We will also consider stability notions that capture resistance to
deviations by individual players. Consider a player $i \in N$ and a pair of coalitions 
$X\not\in\calN(i)$, $Y\in\calN(i)$. 
A player $i$ {\it wants to deviate} from $Y$ to $X$ if $X \cup \{i\} \succ_{i} Y$. 
A player $j \in X$ {\it accepts} a deviation of $i$ to $X$ 
if $X\cup \{i\} \succeq_{j} X$. 
A deviation of $i$ from $Y$ to $X$ is 

\begin{itemize}
\item an {\it NS-deviation} if $i$ wants to deviate from $Y$ to $X$.
\item an {\it IS-deviation} if it is an NS-deviation and all players in $X$ accept it.
\end{itemize}


\noindent A partition $\pi$ is called {\em Nash stable} (NS) (respectively, {\em individually stable} (IS)) 
if no player $i\in N$ has an NS-deviation
(respectively, an IS-deviation) from $\pi(i)$ to another coalition
$X\in \pi$ or to $\emptyset$.

We have the following containment relations among these classes of outcomes:
SCR$\,\subseteq\,$CR, SCR$\,\subseteq\,$IS, NS$\,\subseteq\,$IS. However, a core stable outcome
need not be individually stable, and an individually stable outcome may fail to be in the core.


\smallskip

\noindent{\bf Graphs and digraphs\ }
An {\em undirected graph}, or simply a {\em graph}, is a pair $(N,L)$, 
where $N$ is a finite set of {\em nodes} and 
$L\subseteq \{\, \{i, j\}\mid i, j\in N,i \neq j \,\}$ is a collection of {\em edges} between nodes. 
Given a set of nodes $X$, the {\em subgraph of $(N, L)$ induced by $X$}
is the graph $(X, L_X)$, where $L_X=\{\{i, j\}\in L\mid i, j\in X\}$.

For a graph $(N,L)$, a sequence of distinct nodes $(i_1, i_2, \ldots, i_k)$, $k\geq 2$, 
is called a {\it path} in $L$ if $\{i_h,i_{h+1}\} \in L$ for $h=1,2,\ldots,k-1$. 
A path $(i_1, i_2, \ldots, i_{k})$, $k\geq 3$, is said to be a {\it cycle} in $L$ if $\{i_{k},i_{1}\} \in L$.
A graph $(N,L)$ is said to be a {\it forest} if it contains no cycles.
A subset $X \subseteq N$ is said to be {\it connected in $(N, L)$} if for every pair of distinct nodes 
$i, j \in X$ there is a path between $i$ and $j$ in $L_X$. 
The collection of all connected subsets of $N$ in $(N,L)$ is denoted by $\calF_{L}$; 
also, we write $\calF_{L}(i)=\calF_L\cap \calN(i)$. 
By convention, we assume that $\emptyset \not \in \calF_{L}$.
A forest $(N,L)$ is said to be a {\it tree} if $N$ is connected in $(N, L)$.
A tree $(N,L)$ is called a {\it star} if there exists a central node $s \in N$ 
such that $L=\{\, \{s, j\} \mid j \in N\setminus\{s\} \,\}$.
A subset $X\subseteq N$ of a graph $(N,L)$ is said to be a {\it clique} if 
for every pair of distinct nodes $i, j\in X$ we have $\{i, j\}\in L$.
\par
A {\it directed graph}, or a {\it digraph}, is a pair $(N, A)$ where $N$ is a finite set of nodes 
and $A \subseteq N \times N$. The elements of $A$ are called the {\it arcs}. 
A sequence of distinct nodes $(i_1, i_2, \ldots, i_k)$, $k\geq 2$, is called a {\it directed 
path} in $A$ if $(i_h,i_{h+1}) \in A$ for $h=1,2,\ldots,k-1$.
Given a digraph $(N, A)$, let $L(A)=\{\, \{i,j\} \mid (i,j) \in A \,\}$: 
the graph $(N, L(A))$ is the {\em undirected version} of $(N, A)$. 
A digraph $(N, A)$ is said to be a {\it rooted tree} if $(N,L(A))$ 
is a tree and each node has at most one arc entering it. A rooted tree has exactly one node 
that no arc enters, called the {\it root}, and there exists a unique directed path from the root to every node of $N$. 
\par
Let $(N, A)$ be a rooted tree. We say that a node $j \in N$ is a {\it parent} of $i$ in $A$ if $(j,i) \in A$. We 
denote by $\pr(i,A)$ the unique parent of $i$ in $A$. A node $j \in N$ is called a {\it successor} of $i$ in $A$ if 
there exists a directed path from $i$ to $j$ in $A$.  We write
\[
\suc(i,A)=\{i\}\cup\{\, j \in N \mid j~\mbox{is a successor of}~i~\mbox{in}~A \,\}.
\]
A node $i \in N$ is called a {\it child} of $X \subseteq N$ in $A$ if $i \not \in X$ and $\pr(i,A) \in X$. 
We write 
\[
\ch(X,A)=\{\, i \in N \mid i \not \in X ~\mbox{and}~ \pr(i,A) \in X \,\}.
\]
The {\it height} of a node $i \in N$ of $(N,A)$ is defined inductively as follows: 
\[
\height(i,A):=
\begin{cases}
0~~\mbox{if $\suc(i,A)=\{i\}$},\\
1+\max \{\, \height(j,A)\mid j \in \suc(i,A)\setminus \{i\}\,\}\\
\quad~\mbox{otherwise}.
\end{cases}
\]

\section{Our model}\label{sec:model}
\noindent
The goal of this paper is to study hedonic games
where agent communication is constrained by a graph.

\begin{definition}
A {\em hedonic game with graph structure}, or a {\em hedonic graph game}, is a 
triple $(N,(\succeq_{i})_{i \in N},L)$ where $(N,(\succeq_{i})_{i \in N})$ is a hedonic game, and 
$L \subseteq \{\, \{i,j\} \mid i, j \in N,i \neq j \,\}$ is the set of communication links between players. 
A coalition $X \subseteq N$ is said to be {\em feasible} if it is connected in $(N,L)$.
\end{definition}
If $(N,L)$ is a clique, a hedonic graph game $(N,(\succeq_{i})_{i \in N},L)$ is equivalent to 
the ordinary hedonic game $(N,(\succeq_{i})_{i \in N})$.

A partition $\pi$ of $N$ is said to be {\em feasible} if $\pi \subseteq \calF_{L}$.
An {\em outcome} of a hedonic graph game is a feasible partition.
The standard definitions of stability concepts (see Section~\ref{sec:prelim})
can be adapted to graph games in a straightforward manner.
Specifically, we say that a coalitional deviation is {\em feasible}
if the deviating coalition itself is feasible; an individual deviation
where player $i$ joins a coalition $X$ is {\em feasible} if $X\cup\{i\}$
is feasible. Now, we modify the definitions in Section~\ref{sec:prelim}
by only requiring stability against feasible deviations. 


We use the notation $(N,U,L)$ to denote an additively separable graph game with 
utility matrix $U:N\times N \rightarrow \bbR$.

\begin{example}\label{ex:Parliament}
{\em
Consider the coalition formation problem in a parliament consisting of three parties:
left-wing ($\sfl$), centrist ($\sfc$), and right-wing ($\sfr$). Then
$\sfl$ and $\sfr$ cannot form a coalition without $\sfc$.
We describe this scenario as an additively separable graph game $(N,U,L)$ where $N=\{\sfl,\sfc,\sfr\}$, 
$L=\{\{\sfl,\sfc\},\{\sfc,\sfr\}\}$, and the utility matrix $U$ is given by
\begin{align*}
&U(\sfl, \sfc)=1, U(\sfl, \sfr)=-2, U(\sfc, \sfl)=2,\\
&U(\sfc, \sfr)=0, U(\sfr, \sfc)= 2, U(\sfr, \sfl)=0.
\end{align*}

The resulting preference profile is as follows: 
\begin{align*}
&\sfl~:~ \{\sfl,\sfc\} \succ_{\sfl} \{\sfl\} \succ_{\sfl}\{\sfl,\sfc,\sfr\} \succ_{\sfl} \{\sfl,\sfr\}\\
&\sfc~:~ \{\sfl,\sfc,\sfr\} \sim_{\sfc} \{\sfl,\sfc\} \succ_{\sfc} \{\sfc,\sfr\} \sim_{\sfc} \{\sfc\}\\
&\sfr~:~ \{\sfl,\sfc,\sfr\} \sim_{\sfr}\{\sfc,\sfr\} \succ_{\sfr} \{\sfl,\sfr\} \sim_{\sfr} \{\sfr\}
\end{align*}
The individually rational feasible partitions of this game are 
$\pi_1=\{\{\sfl,\sfc\},\{\sfr\}\}$, $\pi_2=\{\{\sfl\},\{\sfc,\sfr\}\}$, and $\pi_3=\{\{\sfl\},\{\sfc\},\{\sfr\}\}$. 
The partition $\pi_{1}$ is both core stable and individually stable. 
However, there is no Nash stable partition in this game: 
in $\pi_1$, player $\sfr$ wants to join $\{\sfl,\sfc\}$, while in 
$\pi_2$ and $\pi_{3}$, player $\sfc$ wants to join $\{\sfl\}$.
}
\end{example}
In general, hedonic games may fail to have partitions that are core stable or individually stable~\cite{Banerjee2001,Bogomolnaia2002};
consequently, this is also the case for hedonic graph games on general graphs. 
Therefore, in what follows we mostly focus on graph games where
the underlying graph is acyclic. 

\section{Individual Stability}
\noindent
The main contributions of this section are (1) an efficient algorithm for finding an 
individually stable feasible partition in a hedonic graph game whose underlying graph is a forest;
(2) a proof that in the presence of cycles the existence of an IS 
feasible partition is not guaranteed.
 
\begin{theorem}\label{thm:IS}
Suppose that we are given oracle access to the preference relations $\succeq_i$ of all players
in a hedonic graph game $\calG=(N,(\succeq_{i})_{i \in N},L)$, where $(N,L)$ is a forest.             
Then we can find an
individually stable feasible outcome of $\calG$ in time polynomial in~$|N|$.
\end{theorem}

\begin{proof}
We first give an informal description of our algorithm, followed by pseudocode.
If the input graph $(N, L)$ is a forest, we can process each of its connected components
separately, so we can assume that $(N, L)$ is a tree.
We choose an arbitrary node $r$ to be the root; this transforms $(N, L)$
into a rooted tree $(N,A^{r})$ with root $r$ and determines a hierarchy of players. 
For each player $i$, from the bottom player to the top of the hierarchy, 
we compute a tentative partitioning of the subtree rooted at $i$. 
To this end, among all coalitions that $i$'s children belong to, we identify
those whose members would be willing to let $i$ join them. 
Then we let $i$ choose between his most preferred option 
among all such coalitions and the singleton $\{i\}$.
We then check if any of the successors of $i$ who are adjacent to $i$'s coalition want to join it;
we let them do so if they are approved by the current coalition members.

For a family of subsets $\calP\subseteq \calN(i)$, we denote by $\displaystyle \max_{i}\calP$ the family of the most preferred subsets by $i \in N$ in $\calF$, i.e., 
\[
\max_{i}\calP =\{\, X \in \calP \mid X \succeq_{i} Y~\mbox{for all}~Y \in \calP\,\}.
\]

Given a pair of nonempty subsets $X,Y\subseteq N$, we write $X \msucceq Y$ 
if $X\cap Y \neq \emptyset$ and $X \succeq_{i} Y$ for all $i \in X \cap Y$.

\begin{algorithm}                      
\caption{Finding IS partitions}         
\label{alg:is:general}                          
\begin{algorithmic}[1]                  
\REQUIRE tree $(N,L)$, $r\in N$, oracles for $\succeq_i$, $i\in N$.
\ENSURE $\pi^{(r)}$.
\STATE make a rooted tree $(N,A^{r})$ with root $r$ by orienting all the edges in $L$.
\STATE initialize $\best(i)\leftarrow \emptyset$ and $\pi^{(i)}\leftarrow\emptyset$ for each $i \in N$.
\FOR{$t=0,\ldots,\height(r,A^{r})$}
\FOR{$i \in N$ with $\height(i,A^{r})=t$} 
\STATE $C(i)=\{\, k \in \ch(\{i\},A^{r})\mid \best(k)\cup\{i\} \msucceq \best(k) \,\}$.
\STATE\label{step:chooseb} choose
$\best(i) \in \displaystyle{\max_{i}} (\{\{i\}\} \cup \{\, \best(k)\cup\{i\} \mid k\in C(i)\,\}) $.
\WHILE{there exists $j \in \ch(\best(i),A^{r})$ such that
$\best(i) \cup \{j\} \succ_{j} \best(j)~\mbox{and}~$\\
$\best(i) \cup \{j\} \msucceq \best(i)$}\label{step:criterion}
\STATE $\best(i) \leftarrow \best(i) \cup \{j\}$
\ENDWHILE
\STATE $\pi^{(i)} \leftarrow \{\best(i)\}\cup\{\,\pi^{(k)} \mid k \in \ch(\best(i),A^{r}) \,\}$
\ENDFOR
\ENDFOR
\end{algorithmic}
\end{algorithm}
We will now argue that Algorithm~\ref{alg:is:general} correctly identifies an individually stable partition.
Our argument is based on two lemmas. 

\begin{lemma}\label{lem:guarantee}
For every $i \in N$ and every $k \in \ch(\{i\},A^{r})$, 
if $\best(k)\cup \{i\} \msucceq \best(k)$, then $\best(i) \succeq_{i} \best(k)\cup \{i\}$.
\end{lemma}
Lemma~\ref{lem:guarantee} follows immediately
from the choice of $\best(i)$ in Line~\ref{step:chooseb} and the stopping criterion
of the {\bf while} loop in lines 7--9.

\begin{lemma}\label{lem:IS}
For each $i \in N$, $j \in \suc(i,A^{r})$ and all $X \in \pi^{(i)}\cup \{\emptyset\}$
there is no IS feasible deviation of $j$ from $\pi^{(i)}(j)$ to $X$. 
\end{lemma}

\begin{proof}
We use induction on $\height(i,A^{r})$. For 
$\height(i,A^{r})=0$ our assertion is trivial. 
Suppose that it holds for all $j \in N$ with $\height(j,A^{r}) \leq t-1$,
and consider a player $i$ with $\height(i, A^r)=t$. Assume towards a contradiction 
that there exists an IS feasible deviation of $j \in \suc(i,A^{r})$ from 
$\pi^{(i)}(j)$ to $X \in \pi^{(i)}\cup\{\emptyset\}$, i.e., $X \cup \{j\}$ is connected,
\begin{align}
&\label{eq:b} X\cup\{j\} \succ_{j} \pi^{(i)}(j),~\mbox{and}\\ 
&\label{eq:c} X\cup\{j\} \msucceq X. 
\end{align}
By construction, $\pi^{(i)}$ is individually rational, so $X \neq \emptyset$. If $j \not \in \best(i)$ and 
$X \neq \best(i)$, this would contradict the induction hypothesis. Moreover, by the stopping criteria  
in Line~\ref{step:criterion}, it cannot be the case that $j \not \in \best(i)$ and $X = \best(i)$. Thus, $j \in \best(i)$ and 
$X =\best(k)$ for some $k \in \ch(\{j\},A^{r})$. If $j=i$, Lemma~\ref{lem:guarantee} and~\eqref{eq:c} imply that 
$\pi^{(i)}(i)=\best(i) \succeq_{i} \best(k)\cup \{i\}$, contradicting \eqref{eq:b}. Thus, $j \neq i$.

Suppose that player $j$ joins the coalition $\best(i)$ when $\best(i)$ is initialized in Line~\ref{step:chooseb}. 
Then, $j\in \best(k^*)$ for some $k^{*} \in \ch(\{i\},A^{r})$. It follows that
$\best(k) \in \pi^{(k^{*})}$, since $k \not \in \best(k^{*})$ and $k \in \ch(\{j\},A^{r})$. 
Further, the second stopping criterion of the {\bf while} loop in Line~\ref{step:criterion} ensures 
that $j$'s utility does not decrease during the execution of Algorithm~\ref{alg:is:general}. 
Thus, $\best(i) \succeq_{j} \best(k^*)$. Combining this with~\eqref{eq:b} and~\eqref{eq:c} yields
\[
\best(k)\cup\{j\} \succ_{j} \pi^{(i)}(j)=\best(i) \succeq_{j} \best(k^*)=\pi^{(k^*)}(j),
\]
and $\best(k)\cup\{j\} \msucceq \best(k)$. It follows that $\pi^{(k^{*})}$ 
admits an IS feasible deviation of $j$ from $\pi^{(k^{*})}(j)$ to $\best(k) \in \pi^{(k^{*})}$. 
This contradicts the induction hypothesis. 

On the other hand, suppose that player $j$ joins $\best(i)$ during the {\bf while} loop
in Lines 7--9. Then at that point player $j$ is made better off by 
leaving $\best(j)$ and joining $\best(i)$. From then on, she vetoes all candidates whose presence would make
her worse off. Hence, $\best(i) \succ_j \best(j)$. However, \eqref{eq:c} and 
Lemma~\ref{lem:guarantee} imply that $\best(j) \succeq_j \best(k)\cup\{j\}$. 
Thus, $\best(i) \succ_j \best(k)\cup \{j\}$, contradicting~\eqref{eq:b}.
\end{proof}
The partition $\pi^{(r)}$ is feasible by construction, so applying Lemma~\ref{lem:IS} 
with $i=r$ implies that
$\pi^{(r)}$ is an individually stable feasible partition of $N$.

It remains to analyze the running time of Algorithm~\ref{alg:is:general}.
Consider the execution of the algorithm for a fixed player $i$.
Let $c=|\ch(\{i\},A^{r})|$, $s=|\suc(i,A^{r})|$. Line~5 
requires at most $s$ oracle queries: no successor of $i$ is queried more than once.
Line~6 requires $c$ oracle queries.
Moreover, at each iteration of the {\bf while}
loop in lines 7--9 at least one player joins $\best(i)$, so there are at most 
$s$ iterations, in each iteration we consider at most $s$ candidates,
and for each candidate we perform at most $s$ queries.
Summing over all players, we conclude that the number or oracle
queries is bounded by $O(|N|^4)$.
This completes the proof of the theorem.
\end{proof}


\noindent Theorem~\ref{thm:IS} provides a constructive proof that every hedonic graph game whose
underlying graph $(N, L)$ is a forest admits an individually stable feasible partition.
In contrast, if $(N, L)$ contains a cycle, the players' preferences
can always be chosen so that no individually stable feasible partition exists.

\begin{proposition}\label{prop:tight}
Suppose that the graph $(N, L)$
contains a cycle $C=\{i_{1},i_{2},\ldots,i_{k}\}$ with $k\geq 3$, 
$\{i_{h},i_{h+1}\} \in L$ for $h=1,2,\ldots,k$, where $i_{k+1}:=i_1$. 
Then, we can choose preference relations $(\succeq_{i})_{i \in N}$ 
so that the set of IS feasible partitions of the game 
$(N,(\succeq_{i})_{i \in N},L)$ is empty.
\end{proposition}

\begin{proof}
Let $d$ be the smallest natural number that does not divide $k$. 
For each $i_h \in C$, define
\[
\calC(i_h)=\{\, X \in \calF_{L}(i_h) \mid X\subseteq C, |X|\leq d \,\}.
\]
Notice that $d \geq 2$, and hence $\{i_h\},\{i_h,i_{h+1}\} \in \calC(i_h)$. Define a 
graph game $(N,(\succeq_{i})_{i \in N},L)$ so that for each $i_h \in C$ we have
\begin{enumerate}
\item[$(${\rm i}$)$] $X \succ_{i_h} Y~\mbox{for all}~X,Y \in \calC(i_h)~\mbox{such that}~i_{h+1} \in X \setminus Y$,
\item[$(${\rm ii}$)$] $X \sim_{i_h} Y~\mbox{for all}~X,Y \in \calC(i_h)~\mbox{such that}~i_{h+1} \in X \cap Y$,
\item[$(${\rm iii}$)$] $X \sim_{i_h} Y~\mbox{for all}~X,Y \in \calC(i_h)~\mbox{such that}~i_{h+1} \not \in X\cup Y$, 
\item[$(${\rm iv}$)$] $X \succ_{i_h} Y~\mbox{for all}~X \in \calC(i_h), Y \not \in \calC(i_h)$.
\end{enumerate}

Let $\pi$ be an arbitrary individually rational feasible partition of $N$. 
By individual rationality, $\pi(i) \in \calC(i)$ for every $i \in C$.  
Since $d$ does not divide $k$, there exists $Y \in \pi$ such that $|Y| \leq d-1$ 
and $Y \subseteq C$. Let $Y=\{i_{h+1}, i_{h+2},\ldots,i_{m}\}$. Then we have
$\{i_h\}\cup Y \succ_{i_h} \pi(i_h)$, and $\{i_h\}\cup Y \sim_{i_{j}} Y$ for all $i_{j} \in Y$.
Thus, $\pi$ is not individually stable.  
\end{proof}

We summarize our results for individually stable feasible partitions in the following corollary.

\begin{cor} \label{cor:is_tree}
For the class of hedonic graph games, the following statements are equivalent.
\begin{itemize}
\item[$(${\rm i}$)$]  $(N,L)$ is a forest.
\item[$(${\rm ii}$)$] For every hedonic graph game $(N,(\succeq_{i})_{i \in N},L)$ there exists an individually stable 
feasible partition of $N$.
\end{itemize}
\end{cor}


\section{Core Stability}\label{sec:core}
\noindent
As mentioned in Section~\ref{sec:intro}, classic results by Le Breton et al.~\cite{LeBreton1992}
and Demange~\cite{Demange1994,Demange2004} for non-transferable utility games
imply an analogue of Corollary~\ref{cor:is_tree} for core stable partitions.

\begin{theorem}[\cite{LeBreton1992,Demange1994,Demange2004}]
\label{thm:core_tree}
For the class of hedonic games with graph structure, the following statements are equivalent.
\begin{itemize}
\item[$(${\rm i}$)$] $(N,L)$ is a forest.
\item[$(${\rm ii}$)$] For every hedonic graph game $(N,(\succeq_{i})_{i \in N},L)$ 
there exists a core stable feasible partition of $N$.
\end{itemize}
\end{theorem}

We will now show that these two results can be combined, in the following sense:
if $(N, L)$ is acyclic, then every hedonic game on $(N, L)$ admits a feasible partition
that belongs to the core and is individually stable; moreover, the converse is also true.

\begin{theorem}\label{thm:core_is_tree}
For the class of hedonic games with graph structure, the following statements are equivalent.
\begin{itemize}
\item[$(${\rm i}$)$] $(N,L)$ is a forest.
\item[$(${\rm ii}$)$] For every hedonic graph game $(N,(\succeq_{i})_{i \in N},L)$, 
there exists a feasible partition of $N$ that belongs to the core and is individually stable.
\end{itemize}
\end{theorem}

\begin{proof}
We will argue that if an undirected graph $(N, L)$ 
has the property that every game on $(N, L)$
admits a core stable feasible partition,
then every game on $(N, L)$ admits a feasible partition
that belongs to the core and is individually stable. Combined with Theorem~\ref{thm:core_tree},
this gives the desired result.

Let $(N, L)$ be a graph whose associated game always admits a core stable partition.  
Recall that a preference relation $\succeq_1$ on a set $X$ is a {\em refinement} of a preference
relation $\succeq_2$ on $X$ if for every $a, b\in X$ it holds that $a \succ_2 b$ implies $a\succ_1 b$.

For each $i \in N$, let $>_i$ be a refinement of $\succeq_i$ 
that satisfies $X >_i Y$ whenever $X \sim_i Y$ and $Y\subsetneq X$. 
By supposition, the hedonic graph game $(N,(>_{i})_{i \in N},L)$ 
admits a core stable feasible partition $\pi$.
By construction, $\pi$ is core stable in the original game $(N,(\succeq_{i})_{i \in N},L)$
as well. We will now argue that it is also individually stable. Assume towards a
contradiction that there exists an IS feasible deviation of some player $i \in N$ from $\pi(i)$ to 
$X \in \pi \cup \{\emptyset\}$. That is, 
$X\cup \{i\} \in \calF_{L}$, $X\cup \{i\} \succ_{i} \pi(i)$, and $X\cup\{i\}\succeq_{j} X$
for each $j \in X$. By construction of $(>_{i})_{i \in N}$, this implies that $X \cup \{i\}>_{i}\pi(i)$ and 
$X\cup \{i\}>_j X$ for each $j \in X$. This 
contradicts the fact that $\pi$ is a core stable partition of the game
$(N,(>_{i})_{i \in N},L)$.
\end{proof}

\subsection{Computational complexity of CR}
\noindent
Demange's proof that every hedonic graph game on an acyclic graph admits a core stable outcome 
is constructive: her paper~\cite{Demange2004} provides an algorithm to find 
a core stable partition. This algorithm  
is similar in flavor to Algorithm~\ref{alg:is:general}: it processes the players
starting from the leaves and moving towards the root, calculates 
the ``guarantee level'' of each player, and then 
partitions the players into disjoint groups in such a way that the final outcome satisfies 
their ``guarantee levels''; this is shown to ensure core stability. 
While Demange does not analyze the running time of her algorithm,
it can be verified that it runs in time polynomial in the number of connected
subsets of the underlying graph. Thus, in particular, Demange's algorithm 
runs in polynomial time if this graph is a path.

\begin{theorem}[implicit in~\cite{Demange2004}]
Suppose that we are given oracle access to the preference relations $\succeq_i$ of all players
in a hedonic graph game $\calG=(N,(\succeq_{i})_{i \in N},L)$, where $(N,L)$ is a forest.
Then we can find a core stable feasible outcome of $\calG$ in time polynomial in
the number of connected subsets of $(N, L)$.
\end{theorem}

Combining Demange's algorithm with the construction in the proof of Theorem~\ref{thm:core_is_tree},
we obtain an algorithm that has the same worst-case running time as Demange's algorithm
and outputs a feasible partition that belongs to the core and is individually stable.

\begin{cor}
Suppose that we are given oracle access to the preference relations $\succeq_i$ of all players
in a hedonic graph game $\calG=(N,(\succeq_{i})_{i \in N},L)$, where $(N,L)$ is a forest.
Then we can find a feasible outcome of $\calG$ that belongs to the core and is individually stable
in time polynomial in the number of connected subsets of $(N, L)$.
\end{cor}
However, if the number of connected subsets of $(N, L)$ is super-polynomial in $|N|$, so is the running
time of Demange's algorithm, because for each player $i$ this algorithm considers
all feasible coalitions containing $i$. Now, for many $n$-node trees 
the number of connected subtrees is superpolynomial in $n$: 
for instance, this is the case for every tree with $\omega(\log n)$ leaves,
simply because we can delete any subset of leaves and still obtain a connected graph. Therefore,
it is natural to ask if checking all feasible coalitions is indeed necessary. Note
that when the goal is to find an individually stable partition, the answer to this question is ``no'':
Algorithm~\ref{alg:is:general} only considers some of the feasible coalitions, yet is capable of
finding an individually stable feasible outcome. In contrast, for the core it seems unlikely
that one can obtain a substantial improvement over the running time of Demange's algorithm: 
our next theorem shows that finding a core stable feasible partition of an
additively separable hedonic graph game is NP-hard under Turing reductions even 
if the underlying graph is a star and the input game is a symmetric enemy-oriented game.

\begin{theorem}\label{thm:CR}
If one can find a core stable feasible partition in a symmetric enemy-oriented graph game 
whose underlying graph is a star in time polynomial in the number of players 
then {\em P\,=\,NP}.
\end{theorem}
\begin{proof}
We provide a reduction from the NP-complete {\sc Clique} problem~\cite{gj}.
An instance of {\sc Clique} is a pair $(G, t)$, where
$G$ is an undirected graph and $t$ is a positive integer. It is a ``yes''-instance
if $G$ contains a clique of size at least $t$ and a ``no''-instance otherwise.
We will show how a polynomial-time algorithm for our problem
can be used to decide {\sc Clique} in polynomial time.
 
Given an instance $(G, t)$ of {\sc Clique}, where $G=(V,E)$, we  
construct a symmetric enemy-oriented graph game as follows. We let $N=V\cup\{s\}$
and $L=\{\, \{s,v\} \mid v \in V \,\}$. We will now describe the symmetric matrix $U$. 
Briefly, player $s$ likes all other players and two players in $V$ like each other 
if and only if they are connected by an edge of $G$. Formally, we set
$U(s,v)=1$ for each $v \in V$ and for each $u, v\in V$ we set
$U(u, v)=1$ if $\{u, v\}\in E$ and
$U(u, v)=-|V|-1$ otherwise.

Let $\pi$ be an individually rational feasible partition of this game.
Note that all players in $N\setminus \pi(s)$ form singleton coalitions in $\pi$,
and $\pi(s)\setminus\{s\}$ is a clique in $G$.
We will now argue that $\pi$ is core stable if and only if $\pi(s)\setminus\{s\}$ is a maximum-size
clique in $G$. 

Indeed, if $C$ is a maximum-size clique in $G$
and $|\pi(s)\setminus\{s\}|<|C|$, every player in $C\cup\{s\}$ strictly prefers $C\cup\{s\}$
to its current coalition. Conversely, suppose that $\pi(s)\setminus\{s\}$ is a maximum-size clique,
yet coalition $X$ strictly blocks $\pi$. Then it has to be the case that $s\in X$, and hence
$|X|>|\pi(s)|$; but this means that $X\setminus\{s\}$ is not a clique,
and therefore players in $X\setminus\{s\}$ prefer $\pi$ to $X$, a contradiction.

It follows that, by looking at a core stable feasible outcome $\pi$, 
we can decide whether $G$ contains a clique of size at least $t$.
\end{proof}
Dimitrov et al.~\cite{Dimitrov2006} show that finding a core outcome in
symmetric enemy-oriented games is NP-hard; however, in their model there is no constraint
on communication among the players, i.e., their result is for the case where $(N, L)$
is a clique, whereas our result holds even if $(N, L)$ is a star.

\subsection{Computational complexity of SCR}
\noindent
Unlike core stable outcomes, strictly core stable outcomes need not exist 
even in symmetric enemy-oriented games on stars. Consider, for instance, a variant 
of our parliamentary coalition formation example (Example~\ref{ex:Parliament})
where the centrist party ($\sfc$) is equally happy to collaborate
with the left-wing party ($\sfl$) or the right-wing party ($\sfr$), but the left-wing party
and the right-wing party hate each other. This setting can be captured
by a symmetric enemy-oriented graph game whose underlying graph is a path,
and whose core stable feasible partitions are $\pi_1=\{\{\sfl,\sfc\}, \{\sfr\}\}$
and $\pi_2=\{\{\sfl\},\{\sfc,\sfr\}\}$. However, neither $\pi_1$ nor $\pi_2$
is in the strict core: $\pi_1$ is weakly blocked by $\{\sfc, \sfr\}$
and $\pi_2$ is weakly blocked by $\{\sfl, \sfc\}$.

Our next theorem shows that checking whether a given symmetric enemy-oriented hedonic graph game admits a feasible outcome in the strict core is NP-hard with respect to Turing reductions, even if the underlying graph is a star.

\begin{lemma}
If one can determine whether a graph has a unique maximum-size clique in time polynomial in the number of players then {\em P\,=\,NP}.
\end{lemma}
\begin{proof}
We define {\sc Unique Clique} as the decision problem of determining whether a graph 
has a unique maximum-size 
clique. That is, let $s$ be the size of a maximum clique in $G$;
$G$ is a ``yes''-instance of {\sc Unique Clique} if it contains exactly one clique of size $s$
and a ``no''-instance otherwise. {\sc Clique} admits a Turing reduction 
to {\sc Unique Clique}. Specifically, given an instance $(G,t)$ of 
{\sc Clique} where $G=(V,E)$, we construct $|V|$ instances of {\sc Unique Clique} as follows.
For each $s=1, \dots, |V|$, 
let $C_s$ be a set of size $s$ with $V\cap C_s=\emptyset$, and let  
$H_s = (V \cup C_s, E_s)$, where $\{u, v\}\in E_s$ if and only if $u,v \in C_s$, or 
$u,v \in V$ and $\{u,v\}\in E$. Note that the maximum clique size in $G$
is $r$ if and only if $H_r$ is a ``no''-instance of {\sc Unique Clique}, 
but $H_s$ is a ``yes''-instance of {\sc Unique Clique} for $r=s+1, \dots, |V|$.
Hence, a polynomial-time algorithm for {\sc Unique Clique} can be used to decide {\sc Clique} 
in polynomial time.
\end{proof}

\begin{theorem}\label{thm:SCR:enemy}
If there exists a polynomial-time algorithm that, given 
a symmetric enemy-oriented graph game whose underlying graph is a star, decides whether 
this game has a strictly core stable feasible partition then {\em P\,=\,NP}.
\end{theorem}
\begin{proof}
We provide a reduction from {\sc Unique-Clique}. 
Given an undirected graph $G=(V,E)$, we construct the same symmetric enemy-oriented graph game $(N,U,L)$ as in the proof of Theorem \ref{thm:CR}. We will argue that the $G$ is a ``yes''-instance of {\sc Unique Clique}
if and only if $(N, U, L)$ has a feasible outcome that is strictly core stable.

Suppose that $G$ contains a unique maximum-size clique $C$. Then the 
feasible partition $\pi=\{ C\cup \{s\} \}\cup \{\, \{v\} \mid v \in V \setminus C \,\}$ is strictly core stable. Indeed, $\pi$ is individually rational. Suppose that it is weakly blocked by some coalition $X$. We have $s\in X$ by connectivity, and by individual rationality $X\cap V$ is a clique in $G$. Then, by the unique maximality of $C$, we have $|X \cap V| < |C|$. This implies that $s$ strictly prefers $C$ to $X$. Thus, we obtain a contradiction.
Conversely, suppose that there exists a strictly core stable feasible partition $\pi$ of $N$. If $C$ is a maximum-size clique and $C \neq \pi(s)\setminus \{s\}$, then $C \cup \{s\}$ weakly blocks $\pi$. Hence, $\pi(s)\setminus \{s\}$ is the unique maximum-size clique of $G$.
\end{proof}

Note that in the symmetric enemy-oriented game used for proving Theorem \ref{thm:CR} and Theorem \ref{thm:SCR:enemy}, a strictly core stable feasible partition of the game exists if and only if the core stable partition is unique. Combining this with Theorem \ref{thm:SCR:enemy} yields the following corollary.

\begin{cor}
If there exists a polynomial-time algorithm that, given 
a symmetric enemy-oriented graph game whose underlying graph is a star, decides whether 
this game has a unique core stable feasible partition then {\em P\,=\,NP}.
\end{cor}

For the broader class of symmetric additively separable hedonic graph games on stars,
we obtain an NP-hardness result under the more standard notion of a many-one reduction  
(for games on cliques, this follows from the results of Aziz et al.~\cite{Aziz2013}).

\begin{theorem}\label{thm:SCR:additive}
Given a symmetric additively separable hedonic graph game whose underlying graph is a star,
it is {\em NP}-hard to determine
whether it has a strictly core stable feasible partition.
\end{theorem}
\begin{proof}
Again, we provide a reduction from {\sc Clique}.
Given an undirected graph $G=(V,E)$ and a positive integer $t\ge 2$, 
we construct a symmetric additively separable graph game $(N,U,L)$ where $N=\{a,b,c\} \cup V$, 
$L=\{\{a,b\},\{b,c\}\}\cup\{\,\{b,v\} \mid v \in V \,\}$. Let $M=|N|+1$. 
The utility matrix $U:N\times N \rightarrow \bbR$ is given as 
follows (see Figure~\ref{fig:ex_scr}):
\begin{align*}
&U(a,b)=U(c,b)=t-1, U(a,c)=-M,\\ 
&U(a,v)=U(c,v)=-M,  U(b, v)=1~\mbox{for each}~v\in V,\\
&U(u, v)=-1/(t-1)~\mbox{if}~\{u, v\}\in E\\
&\mbox{and}~U(u,v)=-M~\mbox{otherwise, for all $u,v\in V$}. 
\end{align*}

Suppose that $G$ contains a clique $C$ of size $t$. Then the 
feasible partition $\pi=\{ \{a\},\{c\},C\cup \{b\} \}\cup \{\, \{v\} \mid v \in V \setminus C \,\}$ is 
strictly core stable. Indeed, $\pi$ is individually rational. Suppose that it is weakly blocked
by some coalition $X$. The coalitions $\{a, b\}$ and $\{b,c\}$ are not weakly blocking,
so $X\neq\{a,b\}, \{b,c\}$ and hence by individual rationality $a, c\not\in X$.
Further, we have $b\in X$ by connectivity, and by individual rationality 
$X\cap V$ is a clique in $G$.
If $|X \cap V| < t$, we have $\pi(b) \succ_b X$, and if $|X\cap V|>t$ we have
$\pi(i)\succ_i X$ for each $i\in X\cap V$.
When $|X \cap V|= t$, the utilities of all players in $X$ are the same as in $\pi$
($t$ for $b$, $0$ for other players), so $X$ is not a weakly blocking coalition
in this case either. Thus, we obtain a contradiction.

Conversely, suppose that there exists a strictly core stable feasible partition $\pi$ of $N$. 
We will prove that $\pi(b)\setminus \{b\}$ is a clique of 
size at least $t$ in $G$. If $\pi$ contains $\{a, b\}$, it is weakly blocked by $\{b, c\}$
and vice versa, so $\pi$ contains neither of these two coalitions.
By individual rationality, it is not possible that $\pi(b)=\{a,b,c\}$, or that $\pi(b)\cap 
V \neq \emptyset$ and $\pi(b) \cap \{a,c\}\neq \emptyset$. 
Thus, $\pi(b)\subseteq V\cup\{b\}$.
In order for $\pi$ not to be weakly blocked by $\{a,b\}$, the coalition $\pi(b)$
must contain at least $t$ players from $V$. By individual rationality,
these $t$ players must form a clique in $G$.  
\end{proof}

\begin{figure}
\centering
\begin{tikzpicture}[scale=0.8, transform shape]
	\def \radius {2.2cm}
	\node[draw, circle](b) at (0,0) {$b$};
	\node[draw, circle](c) at ({55}:\radius) {$c$};
	\node[draw, circle](a) at ({125}:\radius) {$a$};
	\node[draw, circle,fill=gray!50](node1) at ({195}:\radius) {$v_{1}$};
	\node[draw, circle,fill=gray!50](node2) at ({232}:\radius) {$v_{2}$};
	\node[draw, circle,fill=gray!50](node3) at ({270}:\radius) {$v_{3}$};
	\node[draw, circle,fill=gray!50](node4) at ({345}:\radius) {$v_{n}$};
    
    \DoubleLine{b}{a}{-,dashed}{-,black,ultra thick};
    \DoubleLine{b}{c}{-,black,ultra thick}{-,dashed};    
    \DoubleLine{node4}{b}{-,black,ultra thick}{-,dashed};
	\DoubleLine{node2}{node3}{-,gray}{-,dashed};	
	\draw[-, >=latex,ultra thick] (b)--(node1);
	\draw[-, >=latex,ultra thick] (b)--(node2);
	\draw[-, >=latex,ultra thick] (b)--(node3);	
	\draw[-, >=latex,dashed] (c)--(a);
	\draw[-, >=latex,dashed] (node2)--(node1);
	\draw[-, >=latex,dashed] (node4)--(c);	
	\draw[-, >=latex,gray] (node2)--(node4);
	\draw[-, >=latex,gray] (node3)--(node4);

	\node[below] at ({25}:{2.5}) {$-M$};
	\node[right] at ({45}:{0.8}) {$t-1$};	
	\node[right] at ({102}:{\radius}) {$-M$};
	\node[left] at ({135}:{0.8}) {$t-1$};
	\node[left] at ({217}:\radius) {$-M$};
	\node[below] at ({247}:\radius) {$-\frac{1}{t-1}$};
	\node[above] at ({350}:{\radius/2}) {$1$};
	
	\draw [dotted,thick] (0.7,-\radius+0.5) arc [radius=2, start angle=290, end angle= 320];
\end{tikzpicture}
\caption{Graph used in the proof of Theorem~\ref{thm:SCR:additive}. Thick black lines represent communication links between players, whereas gray lines stand for edges of the given instance $G$. Values on dashed lines are utilities of players.
\label{fig:ex_scr}}
\end{figure}
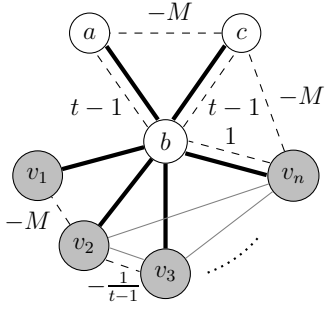

\section{In-Neighbor Stability}\label{sec:INS}
\noindent
In many real-life situations, when people move from one group to another, they need approvals from 
their contacts in the new group. 
Suppose, for instance, that Alice is an early-career researcher 
applying for academic positions in universities: her application 
is unlikely to be accepted if it is rejected by her prospective mentors (even if Alice
expects to collaborate with several other faculty members as well). 

Motivated by these considerations, we will now describe a new notion of stability, which is specific to hedonic
graph games. 
 
\begin{definition}
Given a hedonic graph game $(N,(\succeq_{i})_{i \in N},L)$, we say that $j$ is a {\em 
neighbor} of $i$ if $\{i,j\} \in L$. A feasible deviation of a player $i$ to $X \not\in\calN(i)$ 
is called

\begin{itemize}
\item {\em in-neighbor feasible} if it is NS feasible and accepted by all of $i$'s neighbors in $X$.
\item {\em IR-in-neighbor feasible} if it is in-neighbor feasible and for all $j\in X$ it holds that
$X\cup \{i\}\succeq_j \{j\}$.
\end{itemize}
A feasible partition $\pi$ is called {\em in-neighbor stable (INS)} (respectively,
{\em IR in-neighbor stable (IR-INS)}) if no player $i$
has an in-neighbor feasible deviation (respectively, an IR-in-neighbor feasible deviation) 
from $\pi(i)$ to a coalition $X \in \pi \cup \{\emptyset\}$.
\end{definition}

Note that every INS partition is IR-INS, and each IR-INS partition
is individually stable. However, the converse may not be true:
partition $\pi_1$ in Example~\ref{ex:Parliament} is individually stable, but admits an in-neighbor 
feasible deviation. Indeed, all IR partitions in that example are not in-neighbor stable,
so the existence of in-neighbor stable outcomes is not guaranteed,
even in additively separable games on paths. Note however, that 
$\pi_1$ is IR-in-neighbor stable.

The following example even shows that an individually stable partition is not necessarily IR-in-neighbor stable. 

\begin{example}\label{ex2:Parliament}
{\em
Consider again the coalition formation problem in a parliament consisting of five parties:
extreme-left-wing ($\sfel$), left-wing ($\sfl$), centrist ($\sfc$), right-wing ($\sfr$), and extreme-right-wing ($\sfer$). As in Example \ref{ex:Parliament}, 
players cannot form a coalition without political intermediaries.
We formulate this problem as an additively separable graph game $(N,U,L)$ where $N=\{\sfel,\sfl,\sfc,\sfr,\sfer\}$, 
$L=\{\{\sfel,\sfl\},\{\sfl,\sfc\},\{\sfc,\sfr\},\{\sfr,\sfer\}\}$, and the utility matrix $U$ is given by
\begin{align*}
&U(\sfel, \sfl)=-1, U(\sfel, \sfc)=2, U(\sfel, \sfr)=0, U(\sfel, \sfer)=0,\\
&U(\sfl, \sfel)=0, U(\sfl, \sfc)= 0, U(\sfl, \sfr)=-10, U(\sfl, \sfer)=0,\\
&U(\sfc, \sfel)=-2, U(\sfc, \sfl)=2, U(\sfc, \sfr)=2, U(\sfc, \sfer)=-2,\\
&U(\sfr, \sfel)=0, U(\sfr, \sfl)= -10, U(\sfr, \sfc)=0, , U(\sfr, \sfer)=0,\\
&U(\sfer, \sfel)=0, U(\sfer, \sfl)= 0, U(\sfer, \sfc)=2, , U(\sfer, \sfr)=-1
\end{align*}

The resulting preference profile is as follows: 
\begin{align*}
&\sfel~:~ \{\sfel,\sfl,\sfc\} \succ_{\sfl} \{\sfel\} \\
&\sfl~:~\{\sfel,\sfl,\sfc\} \sim_{\sfl} \{\sfl\} \\
&\sfc~:~ \{\sfl,\sfc\} \sim_{\sfc} \{\sfc,\sfr\} \succ_{\sfc} \{\sfel,\sfl,\sfc\} \sim_{\sfc} \{\sfc,\sfr,\sfer\} \sim_{\sfc} \{\sfc\}\\
&\sfr~:~ \{\sfc,\sfr,\sfer\} \sim_{\sfr} \{\sfr\}\\
&\sfer~:~ \{\sfc,\sfr,\sfer\} \succ_{\sfr} \{\sfer\}
\end{align*}
In the above, we omit the coalitions that are not individually rational. The individually rational partitions of this game are 
$\pi_1=\{\{\sfel\},\{\sfl\},\{\sfc\},\{\sfr\},\{\sfer\}\}$,
$\pi_2=\{\{\sfel\},\{\sfl\},\{\sfc,\sfr,\sfer\}\}$,
$\pi_3=\{\{\sfel,\sfl,\sfc\},\{\sfr\},\{\sfer\}\}$,
$\pi_4=\{\{\sfel\},\{\sfl\},\{\sfc,\sfr\},\{\sfer\}\}$,
and $\pi_5=\{\{\sfel\},\{\sfl,\sfc\},\{\sfr\},\{\sfer\}\}$. 
However, there is no IR-in-neighbor stable partition in this game: in $\pi_1$ and $\pi_{2}$, the deviation of player $\sfc$ to $\{\sfl\}$ is IR-in-neibhbor feasible; in $\pi_3$, the deviation of player $\sfc$ to $\{\sfr\}$ is IR-in-neibhbor feasible; in $\pi_4$, the deviation of player $\sfer$ to $\{\sfc,\sfr\}$ is IR-in-neibhbor feasible; finally, in $\pi_5$, the deviation of player $\sfel$ to $\{\sfl,\sfc\}$ is IR-in-neibhbor feasible.
}
\end{example}

\subsection{Computational complexity of INS}
\noindent
We will now present an algorithm to determine the existence of NS, INS and IR-INS outcomes for games on arbitrary 
acyclic graphs. 

\begin{theorem}\label{thm:NS_NIS_IRNIS}
Suppose that we are given oracle access to the preference relations $\succeq_i$ of all players
in a hedonic graph game $\calG=(N,(\succeq_{i})_{i \in N},L)$, where $(N,L)$ is a forest.
Then we can decide whether $\calG$ admits a Nash stable, in-neighbor stable
or IR-in-neighbor stable feasible outcome (and find one if it exists) in time polynomial in
the number of connected subsets of $(N, L)$.
\end{theorem}
\begin{proof}
Our algorithm is similar to Demange's algorithm for the core~\cite{Demange2004}. 
Again, we assume that $(N,L)$ is a tree. Given a hedonic graph game 
$(N,(\succeq_{i})_{i \in N},L)$, we make a rooted tree 
$(N,A^{r})$ by orienting edges in $L$. Then, for each player $i \in N$ and each 
$X \subseteq \suc(i,A^{r})$ with $X\in\calF_L(i)$, we determine whether there exists a stable partition $\pi$ of 
$\suc(i,A^{r})$ with $X\in\pi$. We set $f(X)=1$ if such a partition exists, and set 
$f(X)=0$ otherwise. To this end, for each $j \in \ch(X,A^{r})$ we try to find a coalition
$X_j\in\calF_L(j)$, $X_j\subseteq \suc(j,A^{r})$ such that no player wants to move across the ``border''
between $X$ and $\ch(X, A^r)$ under the given stability requirement. 
A stable solution exists if and only if $f(X)=1$ for some 
coalition $X\in\calF_L(r)$. 
Algorithm~\ref{alg:INS} describes in detail how $f(X)$ is computed.

\begin{algorithm}                      
\caption{Determining the existence of $\alpha$ feasible partitions, 
where $\alpha \in \{\mbox{NS, INS, IR-INS}\}$
\label{alg:INS}                          
}         
\begin{algorithmic}[1]                  
\REQUIRE tree $(N,L)$, $r\in N$, oracles for $\succeq_i$, $i\in N$.
\ENSURE $f:\calF_{L} \rightarrow \{0,1\}$.
\STATE make a rooted tree $(N,A^{r})$ with root $r$ by orienting all the edges in $L$.
\STATE initialize $f(X)\leftarrow 1$ for $X \in \calF_{L}$.
\FOR{$t=0,\ldots,\height(r,A^{r})$}
\FOR{$i \in N$ with $\height(i,A^{r})=t$} 
\FOR{$X \in \calF_{L}(i)$ such that $X\subseteq\suc(i,A^{r})$}
\IF{$X$ is not individually rational}\label{step:IR}
\STATE $f(X) \leftarrow 0$
\ELSE
\FOR{$j \in \ch(X,A^{r})$}
\IF{for each $X_{j} \in \calF_{L}(j)$ such that $X_{j}\subseteq\suc(j,A^{r})$ and $f(X_{j})=1$ 
the deviation of $j$ from $X_{j}$ to $X$ or the deviation of $\pr(j,A^{r})$ from $X$ to $X_{j}$ is $\alpha$ feasible
}\label{step:alpha}
\STATE $f(X) \leftarrow 0$
\ENDIF
\ENDFOR
\ENDIF
\ENDFOR
\ENDFOR
\ENDFOR
\end{algorithmic}
\end{algorithm}

\begin{lemma}\label{lem:ins}
For each $\alpha \in \{\mbox{NS, INS, IR-INS}\}$, each $i \in N$ and each $X \in \calF_{L}(i)$
such that $X\subseteq\suc(i,A^{r})$ 
we have $f(X)=1$ if and only if there exists an $\alpha$ feasible partition $\pi$ of $\suc(i,A^{r})$ such 
that $X \in \pi$.
\end{lemma}
\begin{proof}
The proof is by induction on $\height(i,A^{r})$. 
The claim is immediate when $\height(i,A^{r})=0$. 
Suppose that it holds for all $j \in N$ with $\height(j,A^{r}) \leq t-1$,
and consider a player $i$ with $\height(i, A^r)=t$. 
Consider an arbitrary $X \in \calF_{L}(i)$ such that $X\subseteq\suc(i,A^{r})$.

Suppose first that $f(X)=1$. Line~\ref{step:IR} ensures that $X$ is individually rational. 
Hence, if $X=\suc(i,A^{r})$, then $X$ is an $\alpha$ feasible partition of $\suc(i,A^{r})$.
Now, suppose that $X \neq \suc(i,A^{r})$, 
i.e., $\ch(X,A^{r}) \neq \emptyset$. Since $f(X)=1$, Line~\ref{step:alpha} 
ensures that for each $j \in \ch(X,A^{r})$
there exists a coalition $X_{j} \in \calF_{L}(j)$ such that
$X_{j}\subseteq \suc(j,A^{r})$, $f(X_{j})=1$ and 
neither the deviation of $j$ from $X_{j}$ to $X$ nor the deviation of $\pr(j,A^{r})$ from $X$ to $X_{j}$ 
is $\alpha$ feasible. By the induction hypothesis, for each $j \in \ch(X,A^{r})$ 
there exists an $\alpha$ feasible partition of $\suc(j,A^{r})$ that contains $X_j$;
combining these partitions with $X$, we obtain an
$\alpha$ feasible partition of $\suc(i,A^{r})$ that contains $X$.

Conversely, if $f(X)=0$, then $X$ is not individually rational, or the condition of the
{\bf if} statement in Line~\ref{step:alpha} is satisfied. 
In either case, there is no $\alpha$ feasible partition of $\suc(i,A^{r})$ containing~$X$.
\end{proof}
Lemma~\ref{lem:ins} immediately implies that the input game admits an $\alpha$ feasible partition
for $\alpha \in \{\mbox{NS, INS, IR-INS}\}$ if and only if $f(X)=1$ for some 
$X \in \calF_{L}(r)$. If this is the case, an $\alpha$ feasible partition can be found using standard
dynamic programming techniques.

It remains to analyze the running time of our algorithm. Let $n=|N|$, $s=|\calF_L|$. 
Algorithm~\ref{alg:INS} considers each coalition $X\in\calF_L$ exactly once.
To check that it is individually rational, it makes at most $n$ oracle calls.
Further, $X$ has at most $n$ children. For each child $j$, the algorithm considers
at most $s$ candidate coalitions $X_j$. To check the conditions in Line~\ref{step:alpha}
for a given $X_j$, we need at most two oracle calls in case of Nash stability and in-neighbor stability
and at most $n$ calls in case of IR-in-neighbor stability. We conclude
that our algorithm performs at most 
$O(n^2s^2)$ oracle calls. 
\end{proof}

The following result shows that we should not hope to obtain a polynomial-time algorithm
for finding in-neighbor stable outcomes, even for additively separable 
hedonic graph games on stars.

\begin{theorem}\label{thm:INS}
Given an additively separable hedonic graph game whose underlying graph is a star,
it is {\em NP}-complete to determine
whether it has an in-neighbor stable feasible partition.
\end{theorem}
\begin{proof}
Clearly, the problem is in NP since in-neighbor stability can be verified in polynomial time. Again, we reduce from {\sc Clique}. Given an undirected graph $G=(V,E)$ and a positive integer $t$, we construct 
an additively separable hedonic graph game $(N,U,L)$ as follows. We set $N=V\cup\{a, b, c\}$, and use 
the same graph $(N, L)$ as in the proof of Theorem~\ref{thm:SCR:additive} (see Figure~\ref{fig:ex_scr}).

Let $M=|N|+1$. We define the utility function $U:N\times N \rightarrow \bbR$ as follows.
\begin{align*}
& U(a, b)=1, U(a, c)=-2, U(b, a) = t, \\
& U(b, c)=0, U(c, a)=0, U(c, b)=2,\\
&U(a,v)=U(c,v)=U(v, a)=U(v, c)=-M~\mbox{for each}~v\in V,\\
&U(b, v)=1,  U(v,b)=0~\mbox{for each}~v\in V,\\
&U(u, v)=0~\mbox{if}~\{u, v\}\in E\\
&\mbox{and}~U(u,v)=-M~\mbox{otherwise, for all $u,v\in V$}.
\end{align*} 

Suppose first that $G$ contains a clique $C$ of size $t$. Let 
$\pi =\{ \{a\},\{c\}, C\cup \{b\} \}\cup \{\, \{v\} \mid v \in V \setminus C \,\}$. 
Clearly, players $a$ and $c$ do not want to deviate to another coalition. 
Also, no player $v \in V$ can profitably deviate. Player 
$b$ does not want to deviate to $\{a\}$, $\{c\}$, or any singleton in $V\setminus C$ because 
$\sum_{v \in C}U(b,v) \geq t \geq 1$. Thus, $\pi$ is in-neighbor stable.

Conversely, suppose that $\pi$ is an in-neighbor stable feasible partition of $N$. Then $b$ is not 
together with $a$ or $c$ since this would cause in-neighbor feasible deviations by at least 
one of the three players. Thus, in $\pi$ player $b$ is grouped together with players in $V$ only. 
Then, by in-neighbor stability $|\pi(b)\cap V|\ge t$, and by individual rationality 
$\pi(b)\cap V$ is a clique in $G$.
\end{proof}

A similar reduction shows that it is hard to find a Nash stable outcome (for additively separable hedonic games with unrestricted communication, this was shown by Sung and Dimitrov~\cite{Sung2010})

\begin{theorem}\label{thm:NS}
Given an additively separable hedonic graph game whose underlying graph is a star,
it is {\em NP}-complete to determine
whether it has a Nash stable feasible partition.
\end{theorem}

In contrast, for any hedonic game on a star, we can construct an 
IR-in-neighbor stable partition efficiently, though it turns out that, for additively separable games on arbitrary trees, it is NP-complete to compute such stable outcomes. 

\begin{proposition}\label{prop:additive:IRINS}
Every hedonic graph game $(N,(\succeq_{i})_{i \in N},L)$ where $(N,L)$ is a star
has an IR-in-neighbor stable partition, and given oracle access to the players'
preference relations, such a partition can be found using $O(|N|^3)$ oracle calls.
\end{proposition}
\begin{proof}
If the central node strictly prefers being on her own, 
rather than being in any coalition of size two, a partition with all singletons is IR-in-neighbor stable. Otherwise, 
choose a favorite two-player coalition of the center, and keep adding players to this coalition 
one by one if this deviation is IR-in-neigh\-bor feasible. 
When no player can be added, the resulting partition is 
IR-in-neigh\-bor stable, since the utility of the central node does not decrease during the execution, and there is no 
player who can IR-in-neighbor deviate to the coalition of the center. 
The bound on the running time is immediate.
\end{proof}

\begin{theorem}\label{thm:IRINS}
Given an additively separable hedonic graph game whose underlying graph is a tree,
it is {\em NP}-complete to determine
whether it has an IR-in-neighbor stable feasible partition.
\end{theorem}
\begin{proof}
The problem is in NP since IR-in-neighbor stability can be verified in polynomial time. Again, we reduce from {\sc Clique}. Given an undirected graph $G=(V,E)$ and a positive integer $t$, we construct 
an additively separable hedonic graph game $(N,U,L)$ as follows. We set $N=V\cup\{a,b,c,d,e\}$, $L=\{\{a,b\},\{b,c\},\{c,d\},\{d,e\}\}\cup \{\, \{c,v\} \mid v \in V\,\}$. Let $M=|N|+1$.
We define the utility function $U:N\times N \rightarrow \bbR$ as follows (see Figure~\ref{fig:irins}).
\begin{align*}
&U(a, b)=-1, U(a, c)=2, U(a, d)=0, U(a, e)=0,\\
&U(b, a)=0, U(b, c)= 0, U(b, d)=-M, U(b, e)=0,\\
&U(c, a)=-t, U(c, b)=t, U(c, d)=t, U(c, e)=-t,\\
&U(d, a)=0, U(d, b)= -M, U(d, c)=0, , U(d, e)=0,\\
&U(e, a)=0, U(e, b)= 0, U(e, c)=2, , U(e, d)=-1\\
&U(a,v)=U(b,v)=U(d,v)=U(e,v)=-M~\mbox{for each}~v\in V,\\
&U(v,a)=U(v,b)=U(v,d)=U(v,e)=-M~\mbox{for each}~v\in V,\\
&U(c, v)=1,  U(v,c)=0~\mbox{for each}~v\in V,\\
&U(u, v)=0~\mbox{if}~\{u, v\}\in E\\
&\mbox{and}~U(u,v)=-M~\mbox{otherwise, for all $u,v\in V$}.
\end{align*}

Note that an individually rational coalitions of this game is either $\{a,b,c\}$, $\{b,c\}$, $\{c,d,e\}$, $\{c,e\}$, or a coalition of the form $C\cup \{c\}$ where $C$ is a clique in $G$, and a singleton of each player.

Suppose first that $G$ contains a clique $C$ of size $t$. Let 
$\pi =\{ \{a\},\{b\},\{d\},\{e\}, C\cup \{c\} \}\cup \{\, \{v\} \mid v \in V \setminus C \,\}$. 
Clearly, players $a$, $b$, $d$ and $e$ have no incentive to deviate to another coalition. 
Also, no player $v \in V$ can profitably deviate. Player 
$c$ does not want to deviate to $\{b\}$, $\{d\}$, or any singleton in $V\setminus C$ because 
$\sum_{v \in C}U(c,v) \geq t \geq 1$. Thus, $\pi$ is IR-in-neighbor stable.

Conversely, suppose that $\pi$ is an IR-in-neighbor stable feasible partition of $N$. Then, $c$ is not 
together with $b$ or $d$. Indeed, if $b \in \pi(c)$, the players $a,b,c,d$, and $e$ would be partitioned into $\{a,b,c\}$ and the singletons $\{d\}$ and $\{e\}$, or $\{b,c\}$ and the singletons $\{a\}$, $\{d\}$, and $\{e\}$ by individual rationality of $\pi$. In either case, this would cause IR-in-neighbor feasible deviations by $c$ or $a$. Similarly, if $d \in \pi(c)$, $\pi$ would admit IR-in-neighbor feasible deviations by $c$ or $e$. Thus, in $\pi$ player $c$ is grouped together with players in $V$ only. Also, $\pi(b)=\{b\}$ and $\pi(d)=\{d\}$ by individual rationality. Then, by IR-in-neighbor stability $|\pi(c)\cap V|\ge t$, and by individual rationality 
$\pi(c)\cap V$ is a clique in $G$.
\end{proof}

\begin{figure}
\centering
\begin{tikzpicture}[scale=0.8, transform shape]
	\def \radius {2.2cm}
	\node[draw, circle](a) at ({150}:1.7*\radius) {$a$};
	\node[draw, circle](b) at ({120}:\radius) {$b$};
	\node[draw, circle](c) at (0,0) {$c$};
	\node[draw, circle](d) at ({60}:\radius) {$d$};
	\node[draw, circle](e) at ({30}:1.7*\radius) {$e$};

	\node[draw, circle,fill=gray!50](node1) at ({195}:\radius) {$v_{1}$};
	\node[draw, circle,fill=gray!50](node2) at ({235}:\radius) {$v_{2}$};
	\node[draw, circle,fill=gray!50](node3) at ({275}:\radius) {$v_{3}$};
	\node[draw, circle,fill=gray!50](node4) at ({345}:\radius) {$v_{n}$};

	\draw[-, >=latex,ultra thick] (a)--(b);
	\draw[-, >=latex,ultra thick] (c)--(b);
	\draw[-, >=latex,ultra thick] (c)--(d);
	\draw[-, >=latex,ultra thick] (d)--(e);
	\draw[-, >=latex,ultra thick] (c)--(node1);
	\draw[-, >=latex,ultra thick] (c)--(node2);
	\draw[-, >=latex,ultra thick] (c)--(node3);
	\draw[-, >=latex,ultra thick] (c)--(node4);

	\draw[-, >=latex,gray] (node2)--(node3);
	\draw[-, >=latex,gray] (node2)--(node4);
	\draw[-, >=latex,gray] (node3)--(node4);

	\node[above] at ({90}:{2}) {$-M$};
	\draw[->, >=latex,dashed] ({70}:\radius)--({110}:\radius);
	\node[below] at ({90}:{1.8}) {$-M$};
	\draw[->, >=latex,dashed] ({110}:2)--({67}:2);
	
	\node[above] at ({136}:2.9) {$0$};
	\draw[->, >=latex,dashed] ({125}:2.55)--({145}:3.6);
	\node[below] at ({138}:2.6) {$-1$};
	\draw[->, >=latex,dashed] ({150}:3.4)--({130}:2.2);

	\node[above] at ({44}:2.9) {$0$};
	\draw[->, >=latex,dashed] ({55}:2.55)--({35}:3.6);
	\node[below] at ({42}:2.6) {$-1$};
	\draw[->, >=latex,dashed] ({30}:3.4)--({50}:2.2);								
				
	\node[left] at (128:1.5) {$t$};
	\draw[->, >=latex,dashed] (-0.32,0.2)--({125}:1.9);
	\node[right] at (112:1.2) {$0$};
	\draw[->, >=latex,dashed] ({115}:1.8)--(-0.025,0.3);

	\node[right] at (52:1.5) {$t$};
	\draw[->, >=latex,dashed] (0.32,0.2)--({55}:1.9);
	\node[left] at (68:1.2) {$0$};
	\draw[->, >=latex,dashed] ({65}:1.8)--(0.025,0.3);

	\node[left] at (163:1.7) {$-t$};
	\draw[->, >=latex,dashed] (-0.6,0.1)--({153}:3.4);
	\node[right] at (148:1.5) {$2$};
	\draw[->, >=latex,dashed] ({150}:3.2)--(-0.35,0.1);

	\node[right] at (17:1.7) {$-t$};
	\draw[->, >=latex,dashed] (0.6,0.1)--({27}:3.4);
	\node[left] at (32:1.5) {$2$};
	\draw[->, >=latex,dashed] ({30}:3.2)--(0.35,0.1);

	\node[left] at ({217}:2.2) {$-M$};
	\draw[<->, >=latex,dashed] (node2)--(node1);

	\node[above] at ({352}:{1.2}) {$1$};
	\draw[->, >=latex,dashed] (0.35,0.07)--({350}:1.8);
	\node[below] at ({335}:{1}) {$0$};
	\draw[->, >=latex,dashed] ({340}:1.75)--(0.25,-0.2);
			
	\node[right] at ({5}:{2.5}) {$-M$};
	\draw[<->, >=latex,dashed] (node4)--({27}:3.5);
	
	\node[below] at ({252}:2.3) {$0$};
	\draw[<->, >=latex,dashed] ({245}:2.3)--({265}:2.3);

	\node[above] at ({87}:3.2) {$0$};
	\draw[<->, >=latex,dashed]({35}:4) arc [radius=6.5, start angle=60, end angle= 120];

	\draw[dotted,thick] (0.7,-\radius+0.5) arc [radius=2, start angle=290, end angle= 320];
\end{tikzpicture}
\caption{Graph used in the proof of Theorem~\ref{thm:IRINS}. Thick black lines represent communication links between players, whereas gray lines stand for edges of the given instance $G$. Values on dashed arcs are utilities of players.
\label{fig:irins}}
\end{figure}
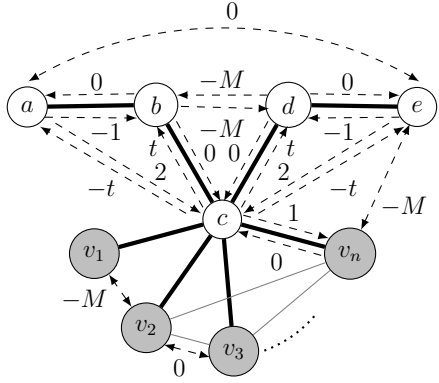

\subsection{Computational complexity of INS:\\ 
Symmetric additively separable games}\label{sec:pls}
\noindent
Note that the additively separable hedonic game used in the proof of Theorem~\ref{thm:NS} is not symmetric (see construction in the proof of Thm.~\ref{thm:INS}).
Indeed, Bogomolnaia and Jackson \cite{Bogomolnaia2002} observe that
in symmetric additively separable hedonic games, any NS deviation strictly increases the sum of all players' 
utilities $\sum_{i \in N}\sum_{j \in \pi(i)}U(i,j)$. This implies that 
for this class of games a sequence of NS deviations
converges to a Nash stable outcome. Thereby, the set of Nash stable outcomes
(and hence also INS, IR-INS, and IS outcomes) is always non-empty.
However, the number of deviations needed to reach a Nash stable outcome
may be exponential in the number of players, so it remains unclear if a
Nash stable outcome can be computed efficiently. 

The complexity class that appears to be useful for capturing the complexity
of this problem is PLS (Polynomial Local Search)~\cite{Johnson1988}. 
A problem in PLS consists of a finite set of candidate solutions, each of which has associated 
neighborhood and cost. It is specified by three polynomial-time algorithms. The first algorithm computes an initial candidate solution (e.g. the all-singleton partition). The second algorithm returns the cost of each candidate solution (e.g. the social welfare of a partition). Finally, the third algorithm tests whether a given candidate solution is optimal in its neighborhood, and if not, finds a solution with better cost (e.g. an improved partition after a profitable deviation).
Given two PLS problems $A$ and $B$, we say that $A$ is {\em PLS-reducible} to $B$ if there exist 
polynomial time computable functions $f$ and $g$ 
such that $f$ maps instances of $A$ to $B$ and $g$ maps the local optima of $B$ to local optima of $A$. 

Gairing and Savani show that
search problems related to NS and IS for symmetric additively separable games are PLS-complete 
\cite{Gairing2010,Gairing2011}. However, if one were to interpret the hedonic game in their reduction
as a graph game, the underlying graph would necessarily contain cycles.
In what follows, we will show that computing in-neighbor stable outcomes for symmetric additively separable games 
is PLS-complete even when the graph $(N,L)$ is a star.

\begin{theorem}\label{thm:sym:INS}
Given a symmetric additively separable hedonic graph game whose underlying graph is a star,
it is {\em PLS}-complete to find an in-neighbor stable feasible partition.
\end{theorem}
\begin{proof}
The candidate solutions of our problem correspond to feasible partitions of the players, and the cost of each partition corresponds to the sum of players' utilities. Two partitions are {\it neighbors} if one can be obtained from the other by moving a single player to another coalition. With symmetric additively separable preferences, we are able to check in polynomial time local optimality of a partition and, if it is not optimal, find a player who can deviate to another coalition and returns the resulting outcome (which has greater social welfare). Hence, our problem is clearly in PLS. 

To prove PLS-hardness, we provide a reduction from {\sc Local Max-Cut}, 
which is known to be PLS-complete \cite{Schaffer1991}. Recall that an instance of {\sc Local Max-Cut}
is given by a weighted graph $G=(V,E,w)$, where $w:V\times V\to{\mathbb N}$ is the weight function with the convention that $w(u,v)=0$ for each $\{u,v\} \not \in E$. 
A {\em cut} is a partition of $V$ into two parts $S$
and $V\setminus S$; its {\em weight} is given by $\sum_{u\in S, v\in V\setminus S}w(u, v)$.
The neighborhood of a cut $(S, V\setminus S)$ is defined as the set of all cuts that can be obtained
by moving one node from $S$ to $V\setminus S$ or vice versa; the goal is to find a cut that has the maximum
weight in its neighborhood. 

Given an instance $(V, E, w)$ of {\sc Local Max-Cut}, we set $N=V\cup\{s\}$,
$L=\{\, \{s,v\} \mid v \in V \,\}$ and construct a symmetric additively separable 
graph game $(N,U,L)$ with the utility matrix $U:N\times N \rightarrow \bbR$ defined as follows. 
We set $U(u,s)=\sum_{v \in V} w(u,v)$ for each $u\in V$. For every pair of distinct nodes 
$u,v \in V$, we define $U(u,v)=-2 w(u,v)$ if $\{u,v\} \in E$ 
and $U(u,v)=0$ otherwise. 
Note that every $u\in V$ will be in a coalition with $s$ or on her own in any feasible partition. Moreover, player $s$ 
will accept all in-neighbor feasible deviations since her utility for every other player is non-negative.

Let $\pi$ be an in-neighbor stable feasible partition of $N$. Let $V_{1}=\{\, u \in V 
\mid u \in \pi(s) \,\}$ and $V_{2}=V\setminus V_1$. 
Each player $u \in V_{1}$ has non-negative utility 
$U(u,s) + \sum_{v \in V_{1}}U(u,v)$ by individual rationality. Hence, $\sum_{v \in V_{1}} 
w(u,v) \leq \sum_{v \in V_{2}} w(u,v)$ for every $u \in V_{1}$. On the other hand, no player 
$u \in V_{2}$ wants to deviate to $\pi(s)$: by a similar calculation, $\sum_{v \in V_{1}} 
w(u,v) \geq \sum_{v \in V_{2}}w(u,v)$ for every $u \in V_{2}$. Thus, $(V_{1},V_{2})$ is a local 
max-cut of $G$.
\end{proof}
The construction in the proof of Theorem~\ref{thm:sym:INS} can also be used to 
show PLS-completeness of finding a Nash stable partition in this class of games.

\begin{theorem}\label{thm:sym:NS}
Given a symmetric additively separable hedonic graph game whose underlying graph is a star,
it is {\em PLS}-complete to find a Nash stable feasible partition.
\end{theorem}
However, we cannot extend Theorem~\ref{thm:sym:NS} to enemy-oriented games.

\begin{proposition}\label{prop:sym:enemy:NS}
A Nash stable feasible outcome of
a symmetric enemy-oriented game on a star can be computed in polynomial time.
\end{proposition}
\begin{proof}
We initialize coalition $X$ to the center of the star;
as long as there is a player that likes (and is liked by) all current members of $X$,
we add him to $X$. Eventually, no player can be added to $X$.
At this point, $\{X\}\cup\{\{i\}\mid i \in N \setminus X\}$
is a Nash stable feasible partition: no player outside of $X$ wants to deviate to $X$,
and no player in $X$ wants to leave.
\end{proof}

\section{Conclusions and Future Work}
\noindent
We have explored the existence and computational complexity of stable partitions in hedonic games 
on acyclic graphs. We obtained a number of algorithmic results in the general oracle-based
framework, thereby showing that acyclicity of the communication network
has important implications for stability.
It remains unknown whether a strictly core stable partition for a hedonic game on a tree
can be computed in time polynomial in the number of connected coalitions and whether an IR-in-neighbor stable partition for a symmetric additively separable game on a tree can be found in time
polynomial in the number of players; we leave these problems as directions for future work.
It would also be interesting to see if our algorithms can be extended to graphs
that are ``almost'' acyclic, and, more broadly,
if there are constraints on the communication
structure other than acyclicity that lead to existence/tractability results for common 
hedonic games stability concepts. 


\bibliographystyle{acm}

\end{document}